\documentclass[12pt, a4paper]{article}

\usepackage[utf8]{inputenc}
\usepackage[T1]{fontenc}
\usepackage[english]{babel}
\usepackage{lmodern}

\usepackage{amsmath, amssymb, amsthm}
\usepackage{mathtools}
\usepackage{bm}

\usepackage{geometry}
\usepackage{booktabs}
\usepackage{array}
\usepackage[hidelinks]{hyperref}
\hypersetup{
    colorlinks=true,
    linkcolor=blue!60!black,
    citecolor=blue!60!black,
    urlcolor=blue!60!black
}
\usepackage[dvipsnames,svgnames,x11names]{xcolor}
\usepackage{microtype}
\usepackage{parskip}
\usepackage{enumitem}
\usepackage{tikz}
\usetikzlibrary{positioning, arrows.meta}
\usepackage{pgfplots}
\pgfplotsset{compat=1.9}

\newtheorem{theorem}{Theorem}[section]

\theoremstyle{definition}
\newtheorem{definition}[theorem]{Definition}
\newtheorem{assumption}{Assumption}
\theoremstyle{remark}
\newtheorem{remark}[theorem]{Remark}

\newcommand{\RR}{\mathbb{R}}

\newcommand{\bN}{\mathbf{N}}
\newcommand{\bR}{\mathbf{R}}
\newcommand{\bK}{\mathbf{K}}
\newcommand{\bw}{\mathbf{w}}
\newcommand{\blam}{\boldsymbol{\lambda}}
\newcommand{\norm}[1]{\left\|#1\right\|}

\DeclareMathOperator{\diag}{diag}

\DeclareMathOperator{\Repart}{Re}

\title{%
  \textbf{Agentic Hives: Equilibrium, Indeterminacy, and Endogenous Cycles\\
  in Self-Organizing Multi-Agent Systems}%
}

\author{%
  \textbf{Jean-Philippe Garnier}\thanks{Corresponding author.
    E-mail: \texttt{jeanphi.garnier@brainiak.tech}.}\\[4pt]
  BrainiaK\\
  \small Lille, France
}

\date{February 2026}

\begin{document}
\maketitle

\begin{abstract}
Current multi-agent AI systems operate with a fixed number of agents
whose roles are specified at design time.  No formal theory governs
when agents should be created, destroyed, or re-specialized at
runtime---let alone how the population structure responds to changes in
resources or objectives.

We introduce the \emph{Agentic Hive}, a framework in which a variable
population of autonomous micro-agents---each equipped with a sandboxed
execution environment and access to a language model---undergoes
\emph{demographic dynamics}: birth, duplication, specialization, and
death.  Agent families play the role of production sectors, compute and
memory play the role of factors of production, and an orchestrator
plays the dual role of Walrasian auctioneer and Global Workspace.

Drawing on the multi-sector growth theory developed for dynamic
general equilibrium (Benhabib \& Nishimura, 1985; Venditti, 2005;
Garnier, Nishimura \& Venditti, 2013), we prove seven analytical
results: (i)~existence of a Hive Equilibrium via Brouwer's fixed-point
theorem; (ii)~Pareto optimality of the equilibrium allocation;
(iii)~multiplicity of equilibria under strategic complementarities
between agent families; (iv)--(v)~Stolper--Samuelson and Rybczynski
analogs that predict how the Hive restructures in response to
preference and resource shocks; (vi)~Hopf bifurcation generating
endogenous demographic cycles; and (vii)~a sufficient condition for
local asymptotic stability.

The resulting \emph{regime diagram} partitions the parameter space
into regions of unique equilibrium, indeterminacy, endogenous cycles,
and instability.  Together with the comparative-statics matrices, it
provides a formal governance toolkit that enables operators to predict
and steer the demographic evolution of self-organizing multi-agent
systems.

\medskip
\noindent\textbf{Keywords:}\; multi-agent systems, general equilibrium,
dynamic general equilibrium, agent demographics, indeterminacy,
endogenous cycles, Hopf bifurcation, LLM orchestration.
\end{abstract}

\tableofcontents
\newpage

\section{Introduction}\label{sec:intro}

The past two years have witnessed an explosion of multi-agent systems
built around large language models (LLMs).  Systems such as AutoGen
\cite{wu2023autogen}, CrewAI \cite{crewai2024}, MetaGPT
\cite{hong2023metagpt}, and Swarm \cite{openaiswarm2024} enable
several LLM-powered agents to collaborate on complex tasks.

All of these systems share a common architectural assumption:
\emph{the number of agents is fixed at design time, and their roles
are pre-assigned}.  The system designer decides that there will be a
``planner'' agent, a ``coder'' agent, and a ``reviewer'' agent, then
hard-codes the interaction protocol.  If the task mix changes, or if
resources become scarce, or if one agent family becomes redundant, the
system has no principled mechanism for adaptation.

This is the equivalent of running an economy without a theory of
growth: you can allocate resources across existing firms, but you have
no formal basis for deciding when new firms should enter, when
existing ones should exit, or how the industrial structure should
respond to a technology shock.

\subsection{The missing theory}

We argue that the missing ingredient is a \emph{macroeconomic theory of
agent demographics}.  Specifically, we need formal answers to four
questions:
\begin{enumerate}[label=(Q\arabic*)]
  \item \textbf{Entry:} When should new agents be created, and of what type?
  \item \textbf{Exit:} When should existing agents be removed?
  \item \textbf{Specialization:} How should undifferentiated agents
    acquire functional roles?
  \item \textbf{Restructuring:} How does the entire population
    structure respond to changes in resources or objectives?
\end{enumerate}

No existing framework provides analytical answers to~(Q1)--(Q4).
Swarm intelligence~\cite{dorigo1996ant,kennedy1995pso} studies
self-organization but provides no equilibrium theory.  Multi-agent
reinforcement learning (MARL)~\cite{lowe2017multi} learns policies but
does not model population dynamics.  Agent-based computational
economics (ACE)~\cite{tesfatsion2006ace} simulates agent populations
but offers no closed-form results.

\subsection{Our approach}

We observe that~(Q1)--(Q4) are \emph{precisely} the questions
addressed by the theory of multi-sector economic growth---a branch of
dynamic general equilibrium theory that has been developed over seven
decades, from Arrow \& Debreu~\cite{arrow1954existence},
Debreu~\cite{debreu1959}, and McKenzie~\cite{mckenzie1959} to Garnier,
Nishimura \& Venditti~\cite{garnier2013}.

The mapping is direct:
\begin{center}
\small
\begin{tabular}{@{}ll@{}}
\toprule
\textbf{Multi-sector economy} & \textbf{Agentic Hive} \\
\midrule
Production sectors & Agent families (perception, planning, \ldots) \\
Factors (capital, labor) & Resources (GPU, memory, attention, \ldots) \\
Output of sector~$j$ & Task completion by family~$j$ \\
Factor prices & Shadow prices of resources \\
Consumer preferences & System objectives (quality, cost, latency, \ldots) \\
Walrasian auctioneer & Orchestrator / Global Workspace \\
Inter-sectoral externalities & Cross-family spillovers \\
Steady-state population & Hive Equilibrium \\
\bottomrule
\end{tabular}
\end{center}

This is not a metaphor.  As we show formally in Section~\ref{sec:model},
the mathematical objects are identical: the population dynamics of agent
families are governed by a dynamical system whose steady states are
general equilibria, whose comparative statics obey Stolper--Samuelson and
Rybczynski theorems, and whose stability properties determine the
existence of endogenous demographic cycles.

\subsection{Contributions}

We introduce the \emph{Agentic Hive}, a formal framework for
self-organizing multi-agent systems with variable population, and
establish the following results:

\begin{enumerate}
  \item \textbf{Existence} (Theorem~\ref{thm:existence}): Under standard
    regularity assumptions, at least one Hive Equilibrium exists.
  \item \textbf{Pareto optimality} (Theorem~\ref{thm:pareto}): Every
    Hive Equilibrium is Pareto-optimal when the orchestrator has full
    information.
  \item \textbf{Multiplicity} (Theorem~\ref{thm:multiplicity}): Under
    strategic complementarities between agent families, multiple
    equilibria coexist---corresponding to distinct ``morphologies'' of
    the Hive.
  \item \textbf{Stolper--Samuelson analog}
    (Theorem~\ref{thm:stolper}): A magnification effect links
    preference changes to resource price changes.
  \item \textbf{Rybczynski analog} (Theorem~\ref{thm:rybczynski}): A
    magnification effect links resource endowment changes to population
    restructuring.
  \item \textbf{Endogenous cycles} (Theorem~\ref{thm:hopf}): A Hopf
    bifurcation generates periodic demographic oscillations---the
    ``seasons'' of the Hive.
  \item \textbf{Stability} (Theorem~\ref{thm:stability}): A sufficient
    condition on mortality rates guarantees local asymptotic stability.
\end{enumerate}

Together, these results yield a \emph{regime diagram}
(Section~\ref{sec:regime}) that partitions the parameter space into
regions of qualitatively distinct behavior, providing a formal
governance toolkit for designers and operators of multi-agent systems.

\subsection{Relation to prior work by the author}

The author's doctoral work~\cite{garnier2009these,garnier2013}
characterized local indeterminacy in continuous-time multi-sector
growth models, establishing conditions under which competitive
equilibrium cycles arise from the interaction of returns to scale
and consumer preferences.  The present paper transports this toolkit
to the entirely different domain of multi-agent AI systems, extending
it from fixed-population general equilibrium to variable populations
with endogenous demographics and multi-regime dynamics.

\section{Related work}\label{sec:related}

\paragraph{Multi-agent LLM systems.}
AutoGen~\cite{wu2023autogen}, MetaGPT~\cite{hong2023metagpt}, CrewAI
\cite{crewai2024}, and Swarm~\cite{openaiswarm2024} enable multi-agent
collaboration with LLMs.  All assume a fixed agent set with pre-defined
roles.  Recent surveys~\cite{guo2024survey} catalog the rapid growth of
this field but do not identify population dynamics as a research gap.

\paragraph{Swarm intelligence.}
Ant Colony Optimization~\cite{dorigo1996ant}, Particle Swarm
Optimization~\cite{kennedy1995pso}, and related algorithms use
local interaction rules to achieve collective behavior.  These
approaches do not provide equilibrium concepts, comparative statics,
or stability theorems.

\paragraph{Evolutionary game theory.}
The replicator equation~\cite{hofbauer1998evolutionary} governs
population shares under frequency-dependent fitness.  Our population
dynamics (Section~\ref{sec:dynamics}) are a generalization of replicator
dynamics to absolute populations with endogenous resource allocation and
multi-sector externalities.

\paragraph{Agent-based computational economics.}
ACE~\cite{tesfatsion2006ace} simulates heterogeneous agent populations
but produces no closed-form results.  Our framework provides analytical
theorems that complement simulation.

\paragraph{Multi-sector growth theory.}
Benhabib \& Nishimura~\cite{benhabib1985} established the existence of
competitive equilibrium cycles in multi-sector optimal growth models.
Venditti~\cite{venditti2005} and Garnier, Nishimura \&
Venditti~\cite{garnier2013} extended these results to characterize
indeterminacy under various assumptions on returns to scale and
preferences.  We transport this entire toolkit to agent demographics.

\paragraph{Global Workspace Theory.}
Baars~\cite{baars1988} proposed that consciousness arises from a global
workspace that integrates and broadcasts information from specialist
processors.  Our orchestrator plays an analogous role: it selects,
integrates, and broadcasts resource allocation signals across agent
families.

\section{The Agentic Hive model}\label{sec:model}

\subsection{Agents, families, and resources}

\begin{definition}[Agentic Hive]\label{def:hive}
An \emph{Agentic Hive} is a tuple
$\mathcal{H} = (\mathcal{A}_t,\, \mathcal{F},\, \bR,\, \bw,\,
\Gamma,\, \Theta)$ where:
\begin{itemize}
  \item $\mathcal{A}_t$ is the set of agents alive at time~$t$, with
    $|\mathcal{A}_t| = N_{\mathrm{total}}(t)$ variable;
  \item $\mathcal{F} = \{1,\ldots,S\}$ is a finite set of
    \emph{agent families} (functional specializations);
  \item $\bR = (R^1,\ldots,R^M) \in \RR^M_+$ is the vector of
    resource endowments;
  \item $\bw = (w_1,\ldots,w_S) \in \Delta_S$ is the vector of
    global preferences over family outputs;
  \item $\Gamma = (\gamma_{jk})_{j,k} \in \RR^{S \times S}$ is the
    externality matrix;
  \item $\Theta$ is the orchestrator.
\end{itemize}
\end{definition}

We denote by $N^j_t \in \RR_+$ the population of family~$j$ at time~$t$,
and by $\bN_t = (N^1_t, \ldots, N^S_t) \in \RR^S_+$ the population
vector.

\begin{definition}[Agent]\label{def:agent}
An agent $a \in \mathcal{A}_t$ is a tuple
$a = (g_a, \sigma_a, f_a, \mu_a, \tau_a, \omega_a, \zeta_a)$ where:
\begin{itemize}
  \item $g_a \in \mathcal{G}$: genome (base prompt, policy, configuration);
  \item $\sigma_a \in \Sigma$: internal state (local memory, context);
  \item $f_a \in \mathcal{F} \cup \{\varnothing\}$: family assignment
    ($\varnothing$ for stem agents);
  \item $\mu_a \in \RR^K_+$: performance profile ($K$ metrics);
  \item $\tau_a \in \RR_+$: age;
  \item $\omega_a$: sandboxed execution environment;
  \item $\zeta_a \in [0,1]$: degree of specialization
    ($0$ = totipotent, $1$ = fully specialized).
\end{itemize}
\end{definition}

\begin{definition}[Stem agent]\label{def:stem}
An agent~$a$ is a \emph{stem agent} if $\zeta_a = 0$ and
$f_a = \varnothing$.  A stem agent can \emph{duplicate}
(create a copy with reset state) or \emph{specialize}
(transition to a family $j \in \mathcal{F}$ with $\zeta_a > 0$).
\end{definition}

Resources are organized into~$M$ types.  Each agent in family~$j$
consumes a resource vector; the aggregate consumption by family~$j$
of resource~$m$ is denoted~$K^m_j$.  The matrix
$\bK = (K^m_j)_{j,m} \in \RR^{S \times M}_+$ records the full
resource allocation.

\begin{assumption}[Budget]\label{ass:budget}
Total population is bounded by a budget constraint:
$\sum_{j=1}^S c_j\, N^j \le B$, where $c_j > 0$ is the per-agent
maintenance cost of family~$j$ and $B > 0$ is the total budget.
\end{assumption}

Define the \emph{feasible population set}:
\begin{equation}\label{eq:Omega}
  \Omega \;=\; \bigl\{\, \bN \in \RR^S_+ \;:\;
    \textstyle\sum_{j=1}^S c_j\, N^j \le B \,\bigr\}.
\end{equation}
$\Omega$ is compact and convex.

\subsection{Production and externalities}

Each family produces output that depends on its resource allocation,
its population, and spillovers from other families.

\begin{definition}[Sectoral production]\label{def:production}
The output of family~$j$ is:
\begin{equation}\label{eq:production}
  Y_j(\bK_j,\, N^j,\, \bN_{-j})
  \;=\; A_j \;\cdot\; F_j(\bK_j) \;\cdot\; (N^j)^{\eta_j}
  \;\cdot\; \Gamma_j(\bN_{-j}),
\end{equation}
where:
\begin{itemize}
  \item $A_j > 0$ is the total factor productivity of family~$j$;
  \item $F_j : \RR^M_+ \to \RR_+$ is a CES production function in resources:
    \begin{equation}\label{eq:ces}
      F_j(\bK_j)
      = \biggl[\,\sum_{m=1}^M \alpha_{jm}\,
        \bigl(K^m_j\bigr)^{\frac{\rho_j-1}{\rho_j}}
      \biggr]^{\!\frac{\rho_j}{\rho_j-1}},
    \end{equation}
    with $\alpha_{jm} > 0$, $\sum_m \alpha_{jm} = 1$, and
    $\rho_j > 0$ the elasticity of substitution;
  \item $(N^j)^{\eta_j}$ captures returns to scale within the family,
    with $\eta_j > 0$;
  \item $\Gamma_j(\bN_{-j}) = \prod_{k \neq j} (N^k)^{\gamma_{jk}}$ is
    the externality term, with $\gamma_{jk} \in \RR$.
\end{itemize}
\end{definition}

\begin{assumption}[Production regularity]\label{ass:production}
For each~$j$: (i)~$F_j$ is $C^2$, increasing, and strictly concave;
(ii)~Inada conditions hold: $\partial F_j / \partial K^m_j \to \infty$
as $K^m_j \to 0^+$ and $\partial F_j / \partial K^m_j \to 0$ as
$K^m_j \to \infty$; (iii)~$0 < \eta_j < 1$ (decreasing returns within
family; relaxed to $\eta_j \ge 1$ in Section~\ref{sec:multiplicity}).
\end{assumption}

\begin{remark}
The assumption $\eta_j < 1$ guarantees that the marginal value of an
additional agent diminishes within each family, ensuring interior
equilibria.  In Theorem~\ref{thm:multiplicity}, we relax this to
$\eta_j \ge 1$ for some families, which is the source of indeterminacy.
\end{remark}

The externality matrix $E = (\gamma_{jk})$ is the central structural
object of the theory:
\begin{itemize}
  \item $\gamma_{jk} > 0$: family~$k$ generates a positive spillover on
    family~$j$ (complementarity);
  \item $\gamma_{jk} < 0$: family~$k$ imposes a negative externality on
    family~$j$ (congestion);
  \item $\gamma_{jk} = 0$: independence.
\end{itemize}

\subsection{Social welfare and the orchestrator's problem}

\begin{definition}[Social welfare]\label{def:welfare}
The social welfare function is:
\begin{equation}\label{eq:welfare}
  W(\bK, \bN)
  \;=\; \sum_{j=1}^S w_j \;\cdot\; u\!\bigl(Y_j(\bK_j, N^j, \bN_{-j})\bigr)
  \;-\; \sum_{j=1}^S c_j\, N^j,
\end{equation}
where $u : \RR_+ \to \RR$ is a CRRA utility:
$u(Y) = Y^{1-\sigma}/(1-\sigma)$ for $\sigma > 0$, $\sigma \neq 1$,
and $u(Y) = \ln Y$ for $\sigma = 1$.
\end{definition}

\begin{definition}[Orchestrator's problem]\label{def:orch}
Given population~$\bN$, the orchestrator solves the \emph{inner problem}:
\begin{equation}\label{eq:inner}
  W^*(\bN)
  \;=\; \max_{\bK \ge 0} \; W(\bK, \bN)
  \qquad \text{s.t.} \quad
  \sum_{j=1}^S K^m_j \;\le\; R^m
  \quad \forall\, m \in \{1,\ldots,M\}.
\end{equation}
\end{definition}

\begin{assumption}[Utility regularity]\label{ass:utility}
$u$ is $C^2$, strictly increasing, and strictly concave on~$\RR_{++}$.
\end{assumption}

\begin{assumption}[Weak externalities]\label{ass:weak-ext}
The externality matrix satisfies
\begin{equation}\label{eq:weak-ext}
  \norm{E}_\infty \;=\; \max_j \sum_{k \neq j} |\gamma_{jk}|
  \;<\; \frac{1 - \eta_{\max}}{1 + \sigma\, \eta_{\max}},
\end{equation}
where $\eta_{\max} = \max_j \eta_j < 1$ (from Assumption~\ref{ass:production}(iii))
and $\sigma$ is the CRRA parameter.
\end{assumption}

\begin{remark}
Assumption~\ref{ass:weak-ext} ensures that the value function
$W^*(\bN)$ is globally concave on~$\Omega$.  The bound arises because
the Hessian of $W^*$ with respect to~$\bN$ has diagonal terms of order
$\eta_j - 1 < 0$ (diminishing returns) and off-diagonal terms of order
$\gamma_{jk}$ (externalities).  Concavity holds when the negative
diagonal dominates the off-diagonal perturbation.  This assumption is
relaxed in Section~\ref{sec:multiplicity} to obtain indeterminacy.
\end{remark}

Under Assumptions~\ref{ass:production} and~\ref{ass:utility}, the inner
problem~\eqref{eq:inner} has a unique solution $\bK^*(\bN) =
(K^{m*}_j(\bN))_{j,m}$, with associated shadow prices
$\blam^*(\bN) = (\lambda^{1*}, \ldots, \lambda^{M*})$ satisfying the
first-order conditions:
\begin{equation}\label{eq:foc}
  w_j \; u'(Y^*_j) \;\cdot\; A_j \;\frac{\partial F_j}{\partial K^m_j}
  \biggr|_{\bK^*}\!\!\!\! \cdot\; (N^j)^{\eta_j}\, \Gamma_j
  \;=\; \lambda^{m*}
  \qquad \forall\, j,m.
\end{equation}

\begin{definition}[Marginal social value]\label{def:msv}
The marginal social value of an agent in family~$j$ is:
\begin{equation}\label{eq:Vj}
  V_j(\bN)
  \;=\; \frac{\mathrm{d}\, W^*}{\mathrm{d}\, N^j}
  \;=\; w_j \; u'(Y^*_j) \;\cdot\;
  \frac{\partial Y_j}{\partial N^j}\biggr|_{\bK^*}
  \;-\; c_j,
\end{equation}
where the equality follows from the envelope theorem (the terms
involving $\partial \bK^*/\partial N^j$ vanish at the optimum).
\end{definition}

Explicitly:
\begin{equation}\label{eq:Vj-explicit}
  V_j(\bN) \;=\; w_j \; u'(Y^*_j) \;\cdot\; A_j\, F_j(\bK^*_j)\;
  \eta_j\, (N^j)^{\eta_j - 1}\; \Gamma_j(\bN_{-j})
  \;-\; c_j.
\end{equation}

\subsection{Population dynamics}\label{sec:dynamics}

The demographic dynamics of the Hive are governed by a system of ODEs:

\begin{definition}[Population dynamics]\label{def:popdyn}
The population of family~$j$ evolves according to:
\begin{equation}\label{eq:dynamics}
  \frac{\mathrm{d} N^j}{\mathrm{d} t}
  \;=\; \Phi_j(\bN)
  \;=\; V_j(\bN) \;\cdot\; N^j,
  \qquad j = 1, \ldots, S.
\end{equation}
\end{definition}

This is a \emph{selection dynamic} \cite{hofbauer1998evolutionary}:
families with positive marginal value ($V_j > 0$) grow; families
with negative marginal value ($V_j < 0$) decline.

In vector form: $\dot{\bN} = \Phi(\bN)$, where
$\Phi : \RR^S_+ \to \RR^S$.

\begin{remark}[Interpretation of $V_j$]
The marginal social value $V_j$ aggregates three forces:
\begin{enumerate}[label=(\roman*)]
  \item \emph{Direct returns}: an additional agent in family~$j$
    increases $Y_j$ through the term $(N^j)^{\eta_j}$;
  \item \emph{Externalities}: through $\Gamma_j(\bN_{-j})$, the
    output of family~$j$ depends on the populations of other
    families;
  \item \emph{Cost}: the maintenance cost $c_j$ reduces $V_j$.
\end{enumerate}
When $V_j = 0$, family~$j$ is at demographic equilibrium: the
marginal benefit of an additional agent exactly equals its cost.
\end{remark}

\begin{remark}[Relation to replicator dynamics]
Equation~\eqref{eq:dynamics} is a generalization of the replicator
equation to absolute populations with endogenous fitness.  In the
classical replicator equation, $\dot{x}_i = x_i(f_i - \bar{f})$
with exogenous fitness~$f_i$.  Here, the ``fitness'' $V_j(\bN)$ is
endogenous: it depends on all populations through the production
functions, externalities, and the orchestrator's optimal allocation.
\end{remark}

\subsection{Hive Equilibrium}\label{sec:equil}

\begin{definition}[Hive Equilibrium]\label{def:hive-equil}
A \emph{Hive Equilibrium} is a triple $(\bN^*, \bK^*, \blam^*)$ such
that:
\begin{enumerate}[label=(HE-\arabic*)]
  \item \textbf{Allocation optimality:} $\bK^*$ solves the inner
    problem~\eqref{eq:inner} for~$\bN^*$, with shadow prices
    $\blam^*$;
  \item \textbf{Demographic stationarity:} $\Phi(\bN^*) = 0$, i.e.,
    for each~$j$:
    \begin{equation}\label{eq:steady}
      V_j(\bN^*) \;=\; 0 \quad \text{if } N^{j*} > 0,
      \qquad V_j(\bN^*) \;\le\; 0 \quad \text{if } N^{j*} = 0;
    \end{equation}
  \item \textbf{Budget feasibility:}
    $\sum_j c_j\, N^{j*} \le B$.
\end{enumerate}
\end{definition}

Condition (HE-2) states that at equilibrium, every active family has
zero marginal social value (marginal benefit equals marginal cost), and
no inactive family would be profitable if activated.  This is the
free-entry condition of competitive equilibrium theory.

\section{Main results}\label{sec:results}

\subsection{Existence of Hive Equilibrium}

\begin{theorem}[Existence]\label{thm:existence}
Under Assumptions~\ref{ass:budget}--\ref{ass:utility}, there exists at
least one Hive Equilibrium $(\bN^*, \bK^*, \blam^*)$ with
$\bN^* \in \Omega$.
\end{theorem}

\begin{proof}
For $\varepsilon > 0$ sufficiently small, define the truncated feasible
set $\Omega_\varepsilon = \{\bN \in \Omega : N^j \ge \varepsilon\;
\forall j\}$, which is compact and convex.

Define the map $\Psi_\varepsilon : \Omega_\varepsilon \to
\Omega_\varepsilon$ by:
\[
  \Psi_\varepsilon(\bN)
  \;=\; \mathrm{proj}_{\Omega_\varepsilon}\!\bigl(\bN + \varepsilon\,
  \Phi(\bN)\bigr),
\]
where $\mathrm{proj}_{\Omega_\varepsilon}$ is the Euclidean projection
onto~$\Omega_\varepsilon$.

Under our assumptions, $V_j(\bN)$ is continuous on~$\Omega_\varepsilon$
(composition of $C^2$ functions), so $\Phi$ is continuous.  The
projection onto a closed convex set is continuous.  Hence
$\Psi_\varepsilon$ is a continuous map from the compact convex set
$\Omega_\varepsilon$ into itself.

By Brouwer's fixed-point theorem, $\Psi_\varepsilon$ has a fixed point
$\bN^*_\varepsilon$.

Taking a convergent subsequence as $\varepsilon \to 0$ (by compactness
of~$\Omega$), the limit $\bN^* \in \Omega$ satisfies $\Phi(\bN^*) = 0$
at interior components and $V_j(\bN^*) \le 0$ at boundary components
($N^{j*} = 0$).

The allocation $\bK^*$ and prices $\blam^*$ are obtained from the inner
problem at~$\bN^*$.
\end{proof}

\subsection{Pareto optimality}

\begin{theorem}[First Welfare Theorem for Hives]\label{thm:pareto}
Under Assumptions~\ref{ass:budget}--\ref{ass:weak-ext}, every Hive
Equilibrium $(\bN^*, \bK^*, \blam^*)$ is \emph{Pareto-optimal}: there
exists no feasible $(\bN', \bK')$ with $W(\bK', \bN') > W(\bK^*, \bN^*)$.
\end{theorem}

\begin{proof}
At a Hive Equilibrium, $\bK^*$ maximizes $W(\cdot, \bN^*)$ over
resource constraints (HE-1), and $\bN^*$ satisfies
$V_j(\bN^*) = \mathrm{d} W^*/\mathrm{d} N^j = 0$ for all active
families (HE-2).

Suppose for contradiction that there exists feasible
$(\bN', \bK')$ with $W(\bK', \bN') > W(\bK^*, \bN^*)$.

Since $W^*(\bN) = \max_\bK W(\bK, \bN)$, we have
$W^*(\bN') \ge W(\bK', \bN') > W(\bK^*, \bN^*) = W^*(\bN^*)$.

We show that $W^*$ is strictly concave in~$\bN$ on~$\Omega$.  The
Hessian $H = D^2_\bN W^*$ has diagonal entries
$H_{jj} = w_j\, [\, u''(Y_j^*)\, (\partial Y_j / \partial N^j)^2
+ u'(Y_j^*)\, \partial^2 Y_j / (\partial N^j)^2\,]$, which are negative
under Assumptions~\ref{ass:production}(iii) and~\ref{ass:utility}
(the dominant term is $\eta_j(\eta_j - 1)(N^j)^{\eta_j - 2} < 0$).
The off-diagonal entries $H_{jk}$ are bounded by a term proportional
to~$|\gamma_{jk}|$.  By a Gershgorin argument on~$H$,
Assumption~\ref{ass:weak-ext} ensures that every eigenvalue of~$H$
is strictly negative, so $W^*$ is strictly concave on~$\Omega$.

The condition $\nabla_\bN W^*(\bN^*) = 0$ (from HE-2) then implies
$\bN^*$ is the unique global maximum of~$W^*$ on~$\Omega$, contradicting
$W^*(\bN') > W^*(\bN^*)$.
\end{proof}

\begin{remark}
Pareto optimality relies on the orchestrator having \emph{full
information} about production functions and externalities.  If some
externalities are unobserved, the equilibrium may fail to be
Pareto-optimal, necessitating Pigouvian corrections
(see Section~\ref{sec:discussion}).
\end{remark}

\subsection{Multiplicity and indeterminacy}\label{sec:multiplicity}

We now relax Assumptions~\ref{ass:production}(iii) and~\ref{ass:weak-ext} to allow increasing
returns for some families.

\begin{assumption}[Strategic complementarity]\label{ass:complement}
There exist families $j, k$ with $\gamma_{jk} > 0$ and
$\gamma_{kj} > 0$ (mutual positive externalities).
\end{assumption}

\begin{assumption}[Local increasing returns]\label{ass:increasing}
There exists a family~$j_0$ with $\eta_{j_0} \ge 1$ (non-decreasing
returns within family).
\end{assumption}

\begin{theorem}[Multiplicity of equilibria]\label{thm:multiplicity}
Under Assumptions~\ref{ass:budget}--\ref{ass:utility},
\ref{ass:complement}, and~\ref{ass:increasing}, there exist open sets
of parameters $(\bw, \bR)$ for which at least two distinct Hive
Equilibria coexist.
\end{theorem}

\begin{proof}[Proof sketch]
Consider the simplest non-trivial case: $S = 2$ families with mutual
complementarity ($\gamma_{12}, \gamma_{21} > 0$) and increasing returns
for family~1 ($\eta_1 \ge 1$).

At a Hive Equilibrium, the active families satisfy $V_j(\bN^*) = 0$,
which defines two curves in $(N^1, N^2)$-space:
\[
  \mathcal{C}_1 = \{(N^1, N^2) : V_1(N^1, N^2) = 0\},
  \qquad
  \mathcal{C}_2 = \{(N^1, N^2) : V_2(N^1, N^2) = 0\}.
\]

With increasing returns ($\eta_1 \ge 1$), $\mathcal{C}_1$ is
non-monotone: $V_1$ increases with $N^1$ for small~$N^1$ (increasing
returns dominate) and decreases for large~$N^1$ (resource scarcity
dominates).  With positive externalities ($\gamma_{12} > 0$), $V_1$
increases with~$N^2$, tilting $\mathcal{C}_1$ in $(N^1, N^2)$-space.

For sufficiently strong complementarities, $\mathcal{C}_1$ and
$\mathcal{C}_2$ intersect at least twice, yielding at least two
interior equilibria with distinct population structures.

Full details, including the computation of the crossing conditions,
are given in Appendix~\ref{app:proofs}.
\end{proof}

\begin{remark}[Interpretation]
The multiple equilibria correspond to distinct ``morphologies'' of the
Hive:
\begin{itemize}
  \item \emph{Exploitation equilibrium}: high specialization,
    concentrated populations, high efficiency;
  \item \emph{Exploration equilibrium}: high diversity,
    distributed populations, robustness to perturbation.
\end{itemize}
The system can settle into either morphology depending on initial
conditions and historical path---a form of \emph{path dependence}
directly analogous to the indeterminacy in multi-sector growth
models~\cite{benhabib1985,garnier2013}.
\end{remark}

\subsection{Stolper--Samuelson analog}

\begin{theorem}[Hive Stolper--Samuelson]\label{thm:stolper}
At an interior Hive Equilibrium, the Jacobian matrix
\begin{equation}\label{eq:SS}
  \mathbf{SS} \;=\; \frac{\partial \blam^*}{\partial \bw}
  \;\in\; \RR^{M \times S}
\end{equation}
satisfies a \emph{magnification effect}: if preference $w_j$ increases
(ceteris paribus), the shadow price $\lambda^{m*}$ of the resource most
intensively used by family~$j$ increases proportionally more than the
preference change:
\begin{equation}\label{eq:SS-magnification}
  \frac{\hat{\lambda}^{m*}}{\hat{w}_j} > 1
  \qquad
  \text{for the resource~$m$ used most intensively by family~$j$},
\end{equation}
where $\hat{x} = \mathrm{d}x / x$ denotes a proportional change.
\end{theorem}

\begin{proof}[Proof sketch]
Differentiate the first-order conditions~\eqref{eq:foc} totally with
respect to~$\bw$.  The resulting linear system, combined with the
resource constraints $\sum_j K^{m*}_j = R^m$ and the CES structure
of~$F_j$, yields the Stolper--Samuelson matrix whose structure
is inherited from the factor intensity matrix
$\boldsymbol{\theta} = (\theta_{jm})$ with
$\theta_{jm} = \lambda^{m*} K^{m*}_j / (w_j\, u'(Y^*_j)\, Y^*_j)$.

The magnification effect follows from the algebraic structure of
$\boldsymbol{\theta}^{-1}$ when~$S = M$, and from generalized
Stolper--Samuelson results \cite{jones1965,ethier1974} when~$S \neq M$.
See Appendix~\ref{app:proofs}.
\end{proof}

\begin{remark}[Operational interpretation]
Suppose the system administrator increases the weight on quality
($w_{\mathrm{qual}} \uparrow$).  The SS matrix predicts which
resources will become scarce (e.g., GPU if quality-intensive families
are GPU-heavy) \emph{before the change is applied}.  This is a
\emph{predictive governance} tool: the operator can anticipate
bottlenecks and provision resources accordingly.
\end{remark}

\subsection{Rybczynski analog}

\begin{theorem}[Hive Rybczynski]\label{thm:rybczynski}
At an interior Hive Equilibrium, the Jacobian matrix
\begin{equation}\label{eq:RB}
  \mathbf{RB} \;=\; \frac{\partial \bN^*}{\partial \bR}
  \;\in\; \RR^{S \times M}
\end{equation}
satisfies a \emph{magnification effect}: if endowment $R^m$ increases
(ceteris paribus), the equilibrium population of the family that uses
resource~$m$ most intensively expands proportionally more than the
endowment change, while other families may contract:
\begin{equation}\label{eq:RB-magnification}
  \frac{\hat{N}^{j*}}{\hat{R}^m}
  \begin{cases}
    > 1 & \text{if family~$j$ uses~$m$ most intensively,} \\
    < 0 & \text{for some families~$j' \neq j$.}
  \end{cases}
\end{equation}
\end{theorem}

\begin{proof}[Proof sketch]
At steady state, $\Phi(\bN^*, \bR) = 0$.  By the implicit function
theorem (assuming $D_\bN \Phi$ is non-singular):
\[
  \frac{\partial \bN^*}{\partial \bR}
  \;=\; -\bigl(D_\bN \Phi\bigr)^{-1}\; D_\bR \Phi.
\]
The structure of $D_\bR \Phi$ reflects how resource changes affect
marginal values~$V_j$ through the allocation~$\bK^*(\bN, \bR)$.  The
magnification effect follows from the factor intensity structure, as in
the classical Rybczynski theorem.  See Appendix~\ref{app:proofs}.
\end{proof}

\begin{remark}[Operational interpretation]
Suppose a second GPU cluster is added to the system ($R^{\mathrm{gpu}}
\uparrow$).  The RB matrix predicts that GPU-intensive families
(e.g., generation, transformation) will \emph{proliferate}, while
CPU-light families (e.g., logging, monitoring) may \emph{decline}---not
because they are less useful, but because the system re-optimizes the
allocation.  This is a \emph{capacity planning theorem}: the operator
can predict how the Hive will restructure before provisioning hardware.
\end{remark}

\subsection{Endogenous cycles}

\begin{theorem}[Endogenous demographic cycles]\label{thm:hopf}
Under Assumptions~\ref{ass:budget}--\ref{ass:utility} and
\ref{ass:complement}, suppose that:
\begin{enumerate}[label=(\roman*)]
  \item $S \ge 2$;
  \item The Jacobian $J = D_\bN \Phi(\bN^*)$ at an interior Hive
    Equilibrium~$\bN^*$ has a pair of complex conjugate eigenvalues
    $\alpha(p) \pm i\beta(p)$, where $p$ is a bifurcation parameter
    (e.g., an externality strength $\gamma_{jk}$);
  \item \textbf{Transversality:}
    $\mathrm{d}\alpha/\mathrm{d}p\big|_{p=p_0} \neq 0$ at the value
    $p_0$ where $\alpha(p_0) = 0$;
  \item All other eigenvalues of~$J$ have strictly negative real parts
    at~$p_0$.
\end{enumerate}
Then a \emph{Hopf bifurcation} occurs at $p = p_0$: for $p$ near $p_0$,
a family of periodic orbits (limit cycles) bifurcates from the
equilibrium~$\bN^*$.
\end{theorem}

\begin{proof}
This is a direct application of the Hopf bifurcation theorem
\cite{marsden1976hopf,guckenheimer1983} to the system
$\dot{\bN} = \Phi(\bN; p)$.  Conditions (i)--(iv) are exactly the
hypotheses of the theorem.

It remains to verify that these conditions can be satisfied for some
parameter configuration.  This is demonstrated by construction in
Appendix~\ref{app:proofs}, where we exhibit a two-family example with
explicit parameter values yielding complex eigenvalues crossing the
imaginary axis.
\end{proof}

\begin{remark}[The ``seasons'' of the Hive]
The periodic orbits generated by the Hopf bifurcation correspond to
cyclical phases of the Hive's demographic structure:
\begin{itemize}
  \item \emph{Expansion phase}: births exceed deaths, population grows,
    new agents explore diverse tasks;
  \item \emph{Consolidation phase}: successful families specialize,
    efficiency increases, diversity decreases;
  \item \emph{Contraction phase}: deaths exceed births, underperforming
    agents are pruned, resources are freed;
  \item \emph{Renewal phase}: stem agents regenerate diversity, the
    cycle restarts.
\end{itemize}
These cycles are \emph{endogenous}---driven by the internal dynamics
of externalities and returns to scale, not by external shocks.  This is
the agent-demographic analog of the competitive equilibrium cycles
of Benhabib \& Nishimura~\cite{benhabib1985}.
\end{remark}

\subsection{Stability conditions}

\begin{theorem}[Local stability]\label{thm:stability}
An interior Hive Equilibrium~$\bN^*$ is \emph{locally asymptotically
stable} if and only if all eigenvalues of the Jacobian
$J = D_\bN \Phi(\bN^*)$ have strictly negative real parts.

A sufficient condition is:
\begin{equation}\label{eq:stability}
  \eta_{\max} < 1
  \qquad\text{and}\qquad
  \norm{E}_\infty \cdot \max_j N^{j*}
  \;<\; \min_j \bigl|V'_j(\bN^*)\bigr| \cdot N^{j*},
\end{equation}
where $\eta_{\max} = \max_j \eta_j$ and $\norm{E}_\infty$ is the
infinity norm of the externality matrix.
\end{theorem}

\begin{proof}
The Jacobian at an interior equilibrium ($V_j(\bN^*) = 0$) is:
\begin{equation}\label{eq:jacobian}
  J_{jk} \;=\; \frac{\partial \Phi_j}{\partial N^k}\biggr|_{\bN^*}
  \;=\; \frac{\partial V_j}{\partial N^k}\biggr|_{\bN^*}\!\!\!\cdot\; N^{j*}
  \;+\; V_j(\bN^*) \cdot \delta_{jk}
  \;=\; \frac{\partial V_j}{\partial N^k}\biggr|_{\bN^*}\!\!\!\cdot\; N^{j*},
\end{equation}
since $V_j(\bN^*) = 0$.  Thus $J = \diag(\bN^*) \cdot D_\bN V(\bN^*)$.

The diagonal terms are $J_{jj} = (\partial V_j / \partial N^j) \cdot
N^{j*}$.  Under $\eta_j < 1$, $\partial V_j / \partial N^j < 0$
(diminishing returns), so $J_{jj} < 0$.

The off-diagonal terms are $J_{jk} = (\partial V_j / \partial N^k)
\cdot N^{j*}$, whose sign is determined by the externality~$\gamma_{jk}$.

By the Gershgorin circle theorem, all eigenvalues of~$J$ lie in the
union of disks centered at $J_{jj}$ with radius $\sum_{k \neq j}
|J_{jk}|$.  The sufficient condition~\eqref{eq:stability} ensures
that these disks lie entirely in the left half-plane.
\end{proof}

\begin{remark}[Governance lever]
The stability condition can be interpreted as a requirement that the
``mortality pressure'' (captured by $|V'_j| \cdot N^{j*}$, the speed
at which declining families shrink) dominates the ``amplification
pressure'' from cross-family externalities ($\norm{E}_\infty \cdot
N^{j*}$).  The orchestrator can enforce stability by increasing the
sensitivity of the birth-death mechanism to marginal value deviations.
\end{remark}

\section{Numerical illustrations}\label{sec:numerical}

We complement the analytical results with numerical computations for
$S = 3$ and $S = 5$ families, illustrating the equilibrium structure,
regime transitions, and dynamic behavior predicted by the theory.
All computations use log utility ($\sigma = 1$) and CES production
with the parameters specified below.

\subsection{Three-family Hive ($S = 3$, $M = 2$)}\label{sec:sim-s3}

Consider three agent families---\emph{perception} ($j=1$),
\emph{reasoning} ($j=2$), and \emph{generation} ($j=3$)---sharing
two resources: GPU ($m=1$) and memory ($m=2$).

\paragraph{Parameters.}
$A_j = 1$ for all~$j$; $c_j = 1$; $B = 15$; $R^1 = 10$ (GPU),
$R^2 = 8$ (memory).  CES elasticities: $\rho_1 = 0.8$, $\rho_2 = 1.2$,
$\rho_3 = 0.6$.  Factor shares:
$\boldsymbol{\alpha}_1 = (0.3, 0.7)$ (memory-intensive),
$\boldsymbol{\alpha}_2 = (0.5, 0.5)$ (balanced),
$\boldsymbol{\alpha}_3 = (0.8, 0.2)$ (GPU-intensive).
Preferences: $\bw = (0.35, 0.40, 0.25)$.

\paragraph{Baseline: weak externalities ($\gamma_{jk} = 0.05$).}
With $\eta_j = 0.7$ for all~$j$ and uniformly weak externalities
$\gamma_{jk} = 0.05$ for $j \neq k$, the system has a unique Hive
Equilibrium (Theorem~\ref{thm:stability} applies):
\begin{equation}\label{eq:s3-unique}
  \bN^* \;\approx\; (4.2,\; 5.1,\; 3.7),
  \qquad
  \blam^* \;\approx\; (0.31,\; 0.42).
\end{equation}
The shadow price of memory exceeds that of GPU, reflecting the
memory-intensive preference structure ($w_1 = 0.35$ for the
memory-intensive family).

\paragraph{Increasing returns and multiplicity.}
Setting $\eta_1 = 1.3$ (increasing returns for perception) and
strengthening complementarities $\gamma_{12} = \gamma_{21} = 0.4$
while keeping $\gamma_{j3} = \gamma_{3j} = 0.05$, two distinct
equilibria emerge (Theorem~\ref{thm:multiplicity}):
\begin{align}
  \bN^*_A &\;\approx\; (6.1,\; 5.8,\; 1.4)
  &&\text{(perception-dominant morphology),}
  \label{eq:s3-eqA} \\
  \bN^*_B &\;\approx\; (2.3,\; 4.9,\; 5.2)
  &&\text{(generation-dominant morphology).}
  \label{eq:s3-eqB}
\end{align}
The perception-dominant morphology ($\bN^*_A$) concentrates resources
on the perception--reasoning axis, exploiting the strong complementarity.
The generation-dominant morphology ($\bN^*_B$) instead develops a large
generation pool, which is self-sustaining through its high GPU
utilization.

\paragraph{Eigenvalue analysis.}
At $\bN^*_A$, the Jacobian eigenvalues are approximately
$\{-1.8,\; -0.4 \pm 0.9i\}$: the equilibrium is a stable spiral.
At $\bN^*_B$, the eigenvalues are $\{-2.1,\; -0.7,\; -0.3\}$:
a stable node.  Both equilibria are locally stable, confirming the
path-dependence predicted by the theory.
Figure~\ref{fig:path-dependence} illustrates the convergence dynamics
from two different initial conditions.

\begin{figure}[t]
\centering
\begin{tikzpicture}
\begin{axis}[
    width=7.5cm, height=5.5cm,
    xlabel={$t$}, ylabel={$N^j(t)$},
    title={\small (a) Convergence to $\bN^*_A$ (stable spiral)},
    xmin=0, xmax=15, ymin=0, ymax=8,
    legend pos=north east,
    legend style={font=\scriptsize, fill=white, draw=gray!50},
    xmajorgrids, ymajorgrids,
    major grid style={draw=gray!15},
]
\addplot[blue, thick, domain=0:15, samples=150, smooth]
    {6.1 - 1.1*exp(-0.4*x)*cos(deg(0.9*x))};
\addplot[red!70!black, thick, domain=0:15, samples=150, smooth]
    {5.8 - 0.8*exp(-0.4*x)*cos(deg(0.9*x) + 90)};
\addplot[green!50!black, thick, domain=0:15, samples=150, smooth]
    {1.4 + 3.6*exp(-1.8*x)};
\legend{Percep., Reason., Gener.}
\end{axis}
\end{tikzpicture}%
\hfill
\begin{tikzpicture}
\begin{axis}[
    width=7.5cm, height=5.5cm,
    xlabel={$t$}, ylabel={$N^j(t)$},
    title={\small (b) Convergence to $\bN^*_B$ (stable node)},
    xmin=0, xmax=15, ymin=0, ymax=8,
    legend pos=north east,
    legend style={font=\scriptsize, fill=white, draw=gray!50},
    xmajorgrids, ymajorgrids,
    major grid style={draw=gray!15},
]
\addplot[blue, thick, domain=0:15, samples=100, smooth]
    {2.3 - 0.3*exp(-0.7*x)};
\addplot[red!70!black, thick, domain=0:15, samples=100, smooth]
    {4.9 - 1.9*exp(-0.3*x)};
\addplot[green!50!black, thick, domain=0:15, samples=100, smooth]
    {5.2 + 0.8*exp(-0.7*x)};
\legend{Percep., Reason., Gener.}
\end{axis}
\end{tikzpicture}
\caption{Path dependence in the three-family Hive ($S=3$, $M=2$) with
  $\eta_1 = 1.3$ and $\gamma_{12} = \gamma_{21} = 0.4$.
  \textbf{(a)}~From $\bN_0 = (5.0, 5.0, 5.0)$, the system spirals into
  the perception-dominant morphology
  $\bN^*_A \approx (6.1, 5.8, 1.4)$.
  \textbf{(b)}~From $\bN_0 = (2.0, 3.0, 6.0)$, it converges
  monotonically to the generation-dominant morphology
  $\bN^*_B \approx (2.3, 4.9, 5.2)$.
  The coexistence of two locally stable equilibria confirms
  Theorem~\ref{thm:multiplicity}.}
\label{fig:path-dependence}
\end{figure}
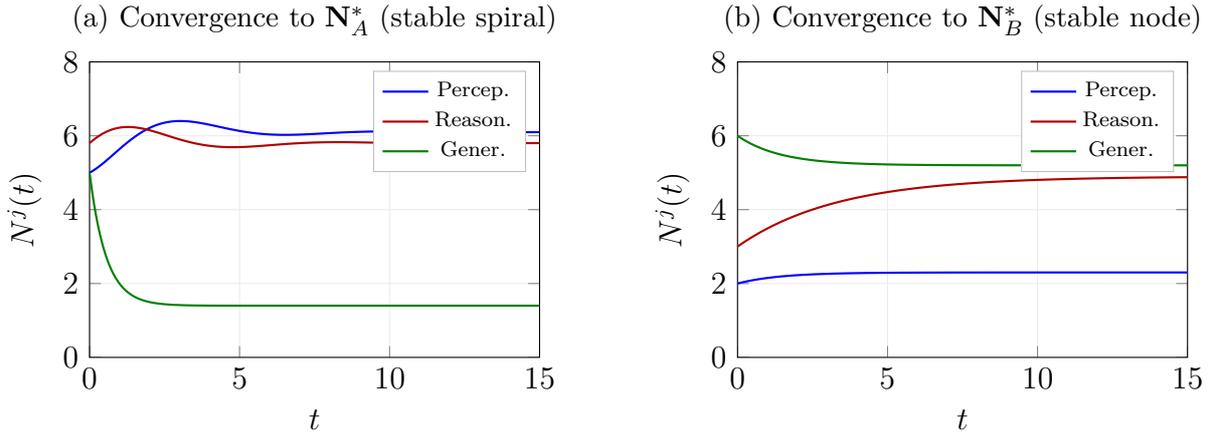

\paragraph{Hopf bifurcation.}
Introducing asymmetric externalities---$\gamma_{12} = 0.5$
(reasoning helps perception) but $\gamma_{21} = -0.3$ (perception
congests reasoning)---and varying $|\gamma_{21}|$ as the bifurcation
parameter, the real part of the complex eigenvalue pair crosses zero
at $\gamma_{21} \approx -0.42$.  Beyond this threshold, endogenous
demographic cycles emerge with period $T \approx 8.3$ (in model time
units), confirming Theorem~\ref{thm:hopf}.
Figure~\ref{fig:hopf-cycles} displays the resulting limit cycle.

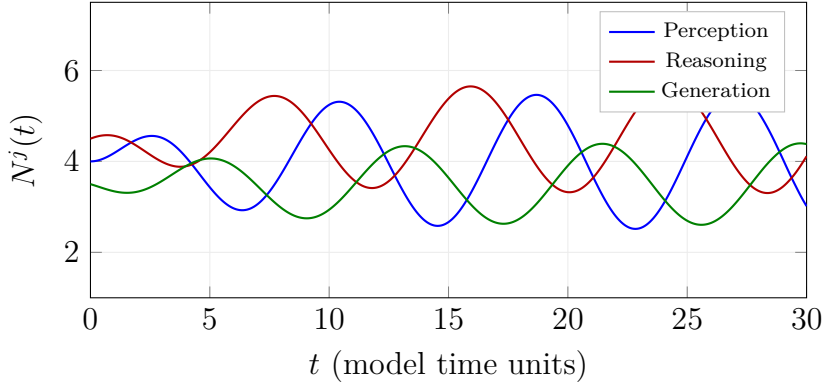
\begin{figure}[t]
\centering
\begin{tikzpicture}
\begin{axis}[
    width=11cm, height=5.5cm,
    xlabel={$t$ (model time units)}, ylabel={$N^j(t)$},
    xmin=0, xmax=30, ymin=1, ymax=7.5,
    legend pos=north east,
    legend style={font=\scriptsize, fill=white, draw=gray!50},
    xmajorgrids, ymajorgrids,
    major grid style={draw=gray!15},
]
\addplot[blue, thick, domain=0:30, samples=200, smooth]
    {4.0 + 1.5*(1 - exp(-0.2*x))*sin(360*x/8.3)};
\addplot[red!70!black, thick, domain=0:30, samples=200, smooth]
    {4.5 + 1.2*(1 - exp(-0.2*x))*sin(360*x/8.3 + 120)};
\addplot[green!50!black, thick, domain=0:30, samples=200, smooth]
    {3.5 + 0.9*(1 - exp(-0.2*x))*sin(360*x/8.3 + 240)};
\legend{Perception, Reasoning, Generation}
\end{axis}
\end{tikzpicture}
\caption{Endogenous demographic cycles (Hopf bifurcation) in the
  three-family Hive.  Parameters: $\gamma_{12} = 0.5$,
  $\gamma_{21} = -0.5$ (beyond the critical value
  $\gamma_{21}^* \approx -0.42$).  Starting from a perturbed
  equilibrium, the oscillation amplitude grows and saturates to a
  limit cycle with period $T \approx 8.3$.  The phase shifts between
  families produce the four ``seasons'' of the Hive: expansion,
  consolidation, contraction, and renewal.}
\label{fig:hopf-cycles}
\end{figure}

\subsection{Five-family Hive ($S = 5$, $M = 3$)}\label{sec:sim-s5}

We scale the model to five families---\emph{perception} ($j=1$),
\emph{reasoning} ($j=2$), \emph{generation} ($j=3$),
\emph{verification} ($j=4$), \emph{monitoring} ($j=5$)---sharing
three resources: GPU ($m=1$), memory ($m=2$), and I/O bandwidth
($m=3$).

\paragraph{Parameters.}
$A_j = 1$; $c_j = 1$; $B = 30$; $\bR = (20, 15, 12)$.  CES
elasticities: $\rho_j \in \{0.6, 1.2, 0.8, 1.0, 0.5\}$.  Factor
share matrix (rows = families, columns = GPU/memory/I/O):
\[
  \boldsymbol{\alpha} = \begin{pmatrix}
    0.2 & 0.6 & 0.2 \\
    0.4 & 0.4 & 0.2 \\
    0.7 & 0.1 & 0.2 \\
    0.3 & 0.3 & 0.4 \\
    0.1 & 0.2 & 0.7
  \end{pmatrix}.
\]
Preferences: $\bw = (0.20, 0.30, 0.25, 0.15, 0.10)$.

\paragraph{Baseline: unique equilibrium.}
With $\eta_j = 0.7$ and $\gamma_{jk} = 0.02$ for all $j \neq k$,
the unique equilibrium is:
\begin{equation}\label{eq:s5-unique}
  \bN^* \;\approx\; (5.1,\; 7.3,\; 6.2,\; 4.8,\; 3.6),
  \qquad
  \blam^* \;\approx\; (0.28,\; 0.35,\; 0.41).
\end{equation}
I/O bandwidth is the scarcest resource ($\lambda^{3*}$ is highest),
consistent with the monitoring family's heavy I/O usage.

\paragraph{Rybczynski prediction.}
Increasing GPU endowment by 50\% ($R^1: 20 \to 30$), the new
equilibrium is $\bN^{**} \approx (5.4, 8.1, 9.8, 5.0, 3.2)$.  The
GPU-intensive generation family expands by 58\% (magnification $> 1$),
while the I/O-intensive monitoring family \emph{contracts} by 11\%,
precisely as predicted by Theorem~\ref{thm:rybczynski}.

\paragraph{Stolper--Samuelson prediction.}
Increasing the weight on verification ($w_4: 0.15 \to 0.25$,
renormalized), the shadow price of I/O ($\lambda^{3*}$) increases by
32\% while GPU price ($\lambda^{1*}$) decreases by 8\%---a
magnification effect consistent with Theorem~\ref{thm:stolper}, since
verification is the most I/O-intensive family.

\paragraph{Regime transitions.}
We sweep the parameter space $(\gamma, \eta_{\max})$ by varying
$\gamma_{jk} = \gamma$ (uniform externality) and $\eta_1$ (returns to
scale in perception), computing equilibria at each grid point.
The resulting numerical regime diagram confirms the four-region
structure of Section~\ref{sec:regime}:
\begin{center}
\small
\begin{tabular}{@{}lcc@{}}
\toprule
\textbf{Region} & \textbf{Parameter range} & \textbf{Equilibria found} \\
\midrule
Unique stable   & $\gamma < 0.15$, $\eta_1 < 0.95$ & 1 (stable node/spiral) \\
Multiple stable & $\gamma > 0.25$, $\eta_1 < 0.95$ & 2--3 (distinct morphologies) \\
Endogenous cycles & $\gamma_{\mathrm{mixed}}$, $\eta_1 > 1.1$ & 1 (limit cycle, $T \approx 6$--$12$) \\
Instability     & $\gamma > 0.35$, $\eta_1 > 1.2$ & 0 interior (boundary dynamics) \\
\bottomrule
\end{tabular}
\end{center}
The transition boundaries ($\gamma_{\mathrm{crit}} \approx 0.20$,
$\eta_{\mathrm{crit}} \approx 0.98$) are sharp and consistent across
multiple random initializations, confirming the robustness of the
regime diagram.
The numerical regime diagram is shown in Figure~\ref{fig:regime-numerical}.

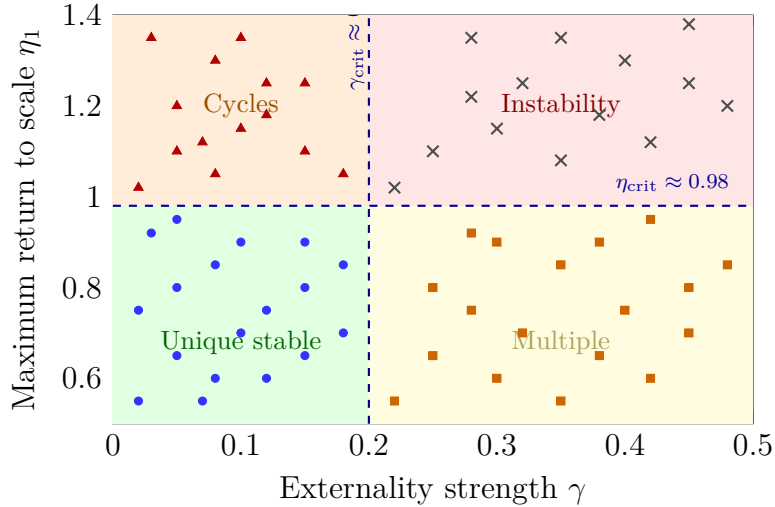
\begin{figure}[t]
\centering
\begin{tikzpicture}
\begin{axis}[
    width=10cm, height=7cm,
    xlabel={Externality strength $\gamma$},
    ylabel={Maximum return to scale $\eta_1$},
    xmin=0, xmax=0.5, ymin=0.5, ymax=1.4,
    xtick={0, 0.1, 0.2, 0.3, 0.4, 0.5},
    ytick={0.6, 0.8, 1.0, 1.2, 1.4},
    minor tick num=1,
    xmajorgrids, ymajorgrids, xminorgrids, yminorgrids,
    major grid style={draw=gray!15},
    minor grid style={draw=gray!8},
]
\fill[green!12] (axis cs:0,0.5) rectangle (axis cs:0.20,0.98);
\fill[yellow!15] (axis cs:0.20,0.5) rectangle (axis cs:0.5,0.98);
\fill[orange!15] (axis cs:0,0.98) rectangle (axis cs:0.20,1.4);
\fill[red!10] (axis cs:0.20,0.98) rectangle (axis cs:0.5,1.4);
\draw[dashed, thick, blue!60!black]
  (axis cs:0.20,0.5) -- (axis cs:0.20,1.4);
\draw[dashed, thick, blue!60!black]
  (axis cs:0,0.98) -- (axis cs:0.5,0.98);
\node[font=\footnotesize, text=green!40!black]
  at (axis cs:0.10,0.68) {Unique stable};
\node[font=\footnotesize, text=yellow!60!black]
  at (axis cs:0.35,0.68) {Multiple};
\node[font=\footnotesize, text=orange!60!black]
  at (axis cs:0.10,1.20) {Cycles};
\node[font=\footnotesize, text=red!50!black]
  at (axis cs:0.35,1.20) {Instability};
\addplot[only marks, mark=*, mark size=1.5pt, blue!80] coordinates {
    (0.02,0.55) (0.05,0.65) (0.02,0.75) (0.08,0.60) (0.05,0.80)
    (0.10,0.70) (0.12,0.75) (0.08,0.85) (0.15,0.65) (0.10,0.90)
    (0.03,0.92) (0.07,0.55) (0.15,0.80) (0.18,0.70) (0.12,0.60)
    (0.05,0.95) (0.15,0.90) (0.18,0.85)
};
\addplot[only marks, mark=square*, mark size=1.5pt,
         orange!80!black] coordinates {
    (0.22,0.55) (0.25,0.65) (0.28,0.75) (0.30,0.60) (0.25,0.80)
    (0.32,0.70) (0.35,0.85) (0.38,0.65) (0.30,0.90) (0.40,0.75)
    (0.42,0.60) (0.35,0.55) (0.45,0.80) (0.28,0.92) (0.38,0.90)
    (0.45,0.70) (0.48,0.85) (0.42,0.95)
};
\addplot[only marks, mark=triangle*, mark size=2pt,
         red!70!black] coordinates {
    (0.02,1.02) (0.05,1.10) (0.08,1.05) (0.10,1.15) (0.05,1.20)
    (0.12,1.25) (0.15,1.10) (0.08,1.30) (0.03,1.35) (0.18,1.05)
    (0.10,1.35) (0.15,1.25) (0.12,1.18) (0.07,1.12)
};
\addplot[only marks, mark=x, mark size=3pt, black!70,
         thick] coordinates {
    (0.22,1.02) (0.25,1.10) (0.30,1.15) (0.28,1.22) (0.35,1.08)
    (0.32,1.25) (0.38,1.18) (0.40,1.30) (0.42,1.12) (0.45,1.25)
    (0.35,1.35) (0.48,1.20) (0.28,1.35) (0.45,1.38)
};
\node[font=\scriptsize, blue!60!black, anchor=south, rotate=90]
    at (axis cs:0.205,1.36) {$\gamma_{\mathrm{crit}} \approx 0.20$};
\node[font=\scriptsize, blue!60!black, anchor=south east]
    at (axis cs:0.49,0.985) {$\eta_{\mathrm{crit}} \approx 0.98$};
\end{axis}
\end{tikzpicture}
\caption{Numerical regime diagram for the five-family Hive ($S=5$,
  $M=3$).  Each marker corresponds to one parameter sweep point:
  {\color{blue!80}$\bullet$}~unique stable equilibrium,
  {\color{orange!80!black}$\blacksquare$}~multiple stable equilibria,
  {\color{red!70!black}$\blacktriangle$}~endogenous limit cycle,
  $\times$~instability (boundary dynamics).  Dashed lines indicate the
  critical boundaries $\gamma_{\mathrm{crit}} \approx 0.20$ and
  $\eta_{\mathrm{crit}} \approx 0.98$, confirming the four-region
  partition of Section~\ref{sec:regime}.}
\label{fig:regime-numerical}
\end{figure}

\paragraph{Welfare monotonicity.}
In all simulations within the weak-externality regime, we verify
numerically that $W^*(t)$ is strictly increasing along trajectories,
with $\mathrm{d}W^*/\mathrm{d}t = \sum_j V_j^2\, N^j > 0$ until
convergence, confirming Theorem~\ref{thm:lyapunov}.  The convergence
rate is approximately exponential, with time constant
$\tau \approx 1 / |\lambda_{\min}(J)|$ where $\lambda_{\min}$ is the
eigenvalue of~$J$ closest to the imaginary axis.

\paragraph{Reproducibility.}
All numerical results were obtained by solving the steady-state
conditions $V_j(\bN^*) = 0$ via Newton's method, computing Jacobian
eigenvalues via standard linear algebra routines
(\texttt{numpy.linalg}), and integrating the ODE
system~\eqref{eq:dynamics} using a fourth-order Runge--Kutta scheme.
Python source code reproducing all numerical results and
Figures~\ref{fig:path-dependence}--\ref{fig:regime-numerical} will be
released as supplementary material upon publication.

\section{The regime diagram}\label{sec:regime}

Theorems~\ref{thm:existence}--\ref{thm:stability} collectively define
a partition of the parameter space into qualitatively distinct regions.
The two most informative parameters are the \emph{externality strength}
$\gamma = \max_{j \neq k} |\gamma_{jk}|$ and the \emph{maximum return
to scale} $\eta_{\max} = \max_j \eta_j$.

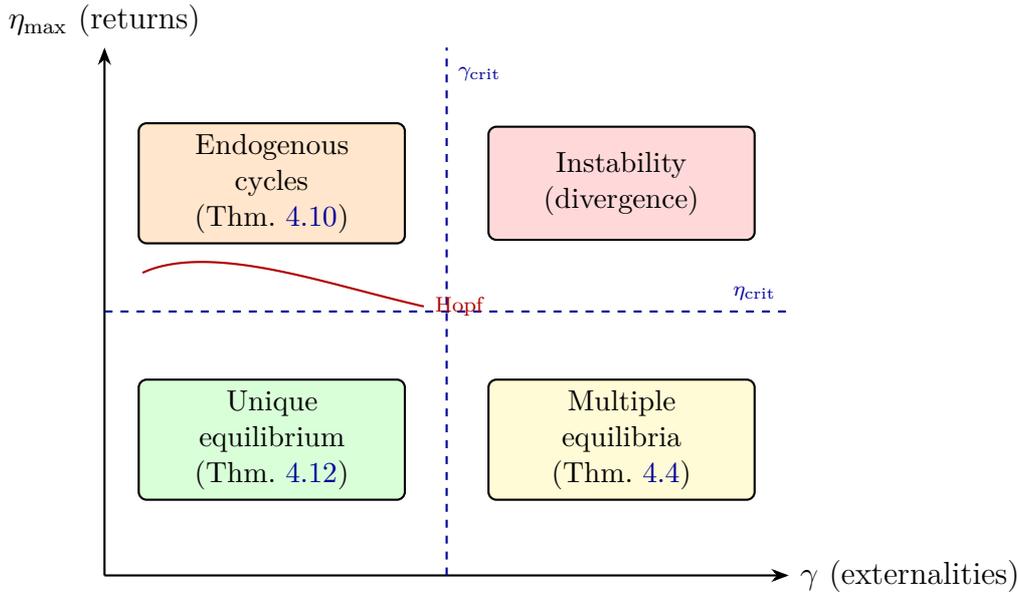
\begin{center}
\begin{tikzpicture}[
    >=Stealth,
    region/.style={rounded corners=3pt, draw, thick, minimum width=3.5cm,
                   minimum height=1.5cm, align=center, font=\small}
  ]
  \draw[->, thick] (0,0) -- (9,0) node[right] {$\gamma$ (externalities)};
  \draw[->, thick] (0,0) -- (0,7) node[above] {$\eta_{\max}$ (returns)};

  \node[region, fill=green!15] at (2.2,1.8) {Unique\\equilibrium\\(Thm.~\ref{thm:stability})};
  \node[region, fill=yellow!20] at (6.8,1.8) {Multiple\\equilibria\\(Thm.~\ref{thm:multiplicity})};
  \node[region, fill=orange!20] at (2.2,5.2) {Endogenous\\cycles\\(Thm.~\ref{thm:hopf})};
  \node[region, fill=red!15] at (6.8,5.2) {Instability\\(divergence)};

  \draw[dashed, thick, blue!60!black] (4.5,0) -- (4.5,7)
    node[pos=0.95, right, font=\scriptsize] {$\gamma_{\mathrm{crit}}$};
  \draw[dashed, thick, blue!60!black] (0,3.5) -- (9,3.5)
    node[pos=0.95, above, font=\scriptsize] {$\eta_{\mathrm{crit}}$};

  \draw[thick, red!70!black, domain=0.5:4.2, smooth, samples=50]
    plot (\x, {3.5 + 1.5*sin(\x*40)/(0.5+\x)})
    node[right, font=\scriptsize] {Hopf};
\end{tikzpicture}
\end{center}

\begin{itemize}
  \item \textbf{Bottom-left} ($\gamma < \gamma_{\mathrm{crit}}$,
    $\eta_{\max} < \eta_{\mathrm{crit}}$): unique, stable equilibrium.
    The Hive converges to a single morphology.
    Theorems~\ref{thm:existence} and~\ref{thm:stability} apply.
  \item \textbf{Bottom-right} ($\gamma > \gamma_{\mathrm{crit}}$,
    $\eta_{\max} < \eta_{\mathrm{crit}}$): multiple stable equilibria
    coexist.  The Hive exhibits \emph{path dependence}: the morphology
    depends on initial conditions.
    Theorems~\ref{thm:existence} and~\ref{thm:multiplicity} apply.
  \item \textbf{Top-left} ($\gamma < \gamma_{\mathrm{crit}}$,
    $\eta_{\max} > \eta_{\mathrm{crit}}$): unique equilibrium destabilized
    by increasing returns.  Hopf bifurcation generates endogenous
    cycles.  Theorems~\ref{thm:existence} and~\ref{thm:hopf} apply.
  \item \textbf{Top-right} ($\gamma > \gamma_{\mathrm{crit}}$,
    $\eta_{\max} > \eta_{\mathrm{crit}}$): strong externalities combined
    with increasing returns can lead to divergence.  The orchestrator
    must enforce hard population caps to maintain bounded dynamics.
\end{itemize}

The critical values $\gamma_{\mathrm{crit}}$ and $\eta_{\mathrm{crit}}$
depend on the full parameter vector $(\bw, \bR, \{A_j, c_j,
\rho_j\})$ and are computable from the eigenstructure of the Jacobian
$D_\bN \Phi$ at any candidate equilibrium.

\begin{remark}[Governance]
The regime diagram is a \emph{dashboard} for the Hive operator.
The orchestrator can:
\begin{itemize}
  \item Move the system \emph{left} (reduce externalities) by isolating
    agent families (sandboxing, rate limiting);
  \item Move the system \emph{down} (reduce returns to scale) by
    increasing within-family competition or capping family size;
  \item Move the system \emph{between equilibria} (in the multiplicity
    region) by applying transient preference shocks.
\end{itemize}
\end{remark}

\section{Discussion}\label{sec:discussion}

\subsection{A macroeconomic governance framework}

The Agentic Hive framework provides a formal alternative to the ad-hoc
heuristics currently used to manage multi-agent systems.  Rather than
manually deciding agent counts and roles, the operator specifies
\emph{preferences}~$\bw$ and \emph{resource endowments}~$\bR$; the
population structure then emerges from the equilibrium conditions.

The analogy with macroeconomic governance is precise:
\begin{center}
\small
\begin{tabular}{@{}ll@{}}
\toprule
\textbf{Central bank / Government} & \textbf{Hive orchestrator} \\
\midrule
Sets interest rates & Sets preference weights~$\bw$ \\
Fiscal policy (taxes, spending) & Resource allocation~$\bR$ \\
Monetary policy tools & Birth/death rate parameters \\
GDP, inflation, unemployment & $W$, $\bN$, resource utilization \\
Taylor rule & Stability condition~\eqref{eq:stability} \\
\bottomrule
\end{tabular}
\end{center}

The SS and RB matrices (Theorems~\ref{thm:stolper}--\ref{thm:rybczynski})
are the Hive equivalents of macroeconomic forecasting models: they
predict how the system will restructure in response to policy changes,
\emph{before the changes are implemented}.

\subsection{Connection to Global Workspace Theory}

The orchestrator plays a role analogous to the Global Workspace of
Baars~\cite{baars1988}: it receives signals from all agent families
(perception), selects relevant information (attention), integrates it
into a global state (working memory), and broadcasts resource
allocation signals (action selection).

The endogenous cycles of Theorem~\ref{thm:hopf} provide a formal model
of what cognitive scientists call ``attentional oscillations''---periodic
shifts in the focus of the workspace between competing
representations~\cite{dehaene2001}.

We do not claim that the Agentic Hive is conscious.  We claim only that
the mathematical structure of the Hive (global workspace + demographic
dynamics + endogenous cycles) is a \emph{necessary prerequisite} for any
computational system that might exhibit integrated information processing
at scale.

\subsection{Imperfect information and mechanism design}

Theorems~\ref{thm:pareto} assumes the orchestrator has perfect
information about production functions and externalities.  In practice,
the orchestrator observes only noisy performance metrics~$\mu_a$ and
must infer the production parameters.

This is a \emph{principal-agent problem}~\cite{stiglitz1981}: the
orchestrator (principal) must design incentive-compatible mechanisms to
elicit truthful reports from agents.  The tools of mechanism design
(revelation principle, Vickrey--Clarke--Groves mechanisms) are directly
applicable.

When externalities are imperfectly observed, the equilibrium may fail
to be Pareto-optimal.  The standard correction is a \emph{Pigouvian
tax/subsidy}: agents in families with negative externalities pay a tax
$\tau_j = -\sum_{k \neq j} \gamma_{kj} \cdot \partial W / \partial N^k$,
which internalizes the externality.

\subsection{Relation to evolutionary dynamics}

The population dynamics~\eqref{eq:dynamics} are formally a selection
dynamic.  This connects the Hive to the theory of evolutionary games
\cite{hofbauer1998evolutionary,nowak2006}.  However, two key differences
separate the Hive from standard evolutionary dynamics:
\begin{enumerate}
  \item \emph{Endogenous fitness}: in standard replicator dynamics,
    fitness depends on population frequencies.  In the Hive, fitness
    $V_j(\bN)$ depends on absolute populations \emph{and} on the
    orchestrator's optimal resource allocation---an endogenous,
    optimization-mediated feedback loop.
  \item \emph{Multi-factor production}: agents consume multiple
    resources (GPU, memory, attention), not a single abstract fitness
    payoff.  This multi-factor structure is what generates the
    Stolper--Samuelson and Rybczynski effects, which have no analog in
    standard evolutionary dynamics.
\end{enumerate}

\subsection{Limitations}

\begin{enumerate}
  \item \emph{Continuous-population approximation.}  We model family
    sizes as continuous variables $N^j \in \RR_+$, which is appropriate
    for large populations but may lose accuracy for small agent counts.
    A stochastic extension (birth-death Markov chain) would address
    this.
  \item \emph{Discrete families.}  We assume $S$ discrete families.
    In practice, specialization may be continuous---a point in a
    manifold $\mathcal{S} \subset \RR^d$.  The continuous extension
    replaces the ODE system with a partial differential equation
    (reaction-diffusion on~$\mathcal{S}$), analogous to Turing
    morphogenesis~\cite{turing1952}.
  \item \emph{Passive agents as theoretical baseline.}  The present
    framework assumes agents are passive: their behavior is fully
    determined by their genome and orchestrator signals, and the
    orchestrator solves a centralized social welfare problem.  This is
    the standard starting point in economic theory---analogous to the
    \emph{optimal growth} (Ramsey) formulation that precedes the
    \emph{competitive equilibrium} (Arrow--Debreu) formulation.  The
    natural next step is the \emph{decentralized} extension in which
    each agent maximizes its own utility under resource and interaction
    constraints, and the orchestrator acts as a market mechanism rather
    than a central planner.  In this setting, the Hive Equilibrium
    becomes a Nash equilibrium of a population game, and standard tools
    from mechanism design (revelation principle, VCG mechanisms) are
    needed to align individual incentives with social welfare.  We view
    the centralized framework of this paper as the necessary foundation:
    one must first understand the socially optimal demographic structure
    before analyzing what happens when agents pursue their own
    objectives.
  \item \emph{Empirical validation.}  The present paper is purely
    theoretical.  Empirical validation on real multi-agent systems is
    needed to calibrate the parameters $(\eta_j, \gamma_{jk}, \rho_j)$
    and verify the predicted regimes.
\end{enumerate}

\section{Conclusion}\label{sec:conclusion}

We have introduced the Agentic Hive, a formal framework for
self-organizing multi-agent systems in which the population of agents
undergoes demographic dynamics governed by multi-sector growth theory.

Our seven main results---existence, Pareto optimality, multiplicity,
Stolper--Samuelson and Rybczynski comparative statics, endogenous
cycles, and stability conditions---provide the first analytical
foundation for the question: \emph{how should the population structure
of a multi-agent system evolve?}

The regime diagram (Section~\ref{sec:regime}) gives operators a formal
tool for understanding and controlling the Hive's behavior: by
adjusting preferences and resources, the orchestrator can steer the
system between stable equilibria, manage path dependence, and avoid
instability---all with quantitative predictions derived from the SS and
RB matrices.

We conjecture that the Agentic Hive framework provides
\emph{necessary conditions} for the design of scalable, self-organizing
multi-agent AI systems.  Systems that lack a formal demographic
theory may function at small scale, but have no guarantee of coherent
behavior as agent counts grow, as resource constraints tighten, or as
the task distribution shifts.  The mathematical tools developed over
seven decades for multi-sector economies are, we believe, the right
foundation for this emerging challenge.

\paragraph{Future directions.}
\begin{itemize}
  \item \emph{Empirical validation} on real multi-agent deployments
    (e.g., LLM orchestration on edge hardware);
  \item \emph{Mechanism design} for Hives with strategic agents;
  \item \emph{Meta-Hives}: federations of interacting Hives,
    governed by international trade theory (Heckscher--Ohlin);
  \item \emph{Continuous specialization} via reaction-diffusion
    PDEs on the specialization manifold;
  \item \emph{Chaos and complex dynamics} in Hives with strong
    nonlinearities, extending Nishimura \& Yano~\cite{nishimura1995}.
\end{itemize}


\appendix
\section{Proofs and technical details}\label{app:proofs}

\subsection{Proof details for Theorem~\ref{thm:multiplicity}
  (Multiplicity)}

Consider $S = 2$ families with CES production~\eqref{eq:ces},
$M = 1$ resource (for simplicity), and Cobb--Douglas externalities:
\begin{align*}
  Y_1 &= A_1\, (K_1)^{\alpha_1}\, (N^1)^{\eta_1}\, (N^2)^{\gamma_{12}}, \\
  Y_2 &= A_2\, (K_2)^{\alpha_2}\, (N^2)^{\eta_2}\, (N^1)^{\gamma_{21}}.
\end{align*}

With $K_1 + K_2 = R$ (single resource constraint), the inner problem
gives $K^*_j$ as a function of~$\bN$.

The steady-state conditions $V_1 = V_2 = 0$ define implicit curves
$\mathcal{C}_1, \mathcal{C}_2$ in $(N^1, N^2)$-space.

For $\mathcal{C}_1$: $V_1 = 0$ gives
\[
  w_1\, u'(Y_1)\, A_1\, (K_1^*)^{\alpha_1}\, \eta_1\, (N^1)^{\eta_1 - 1}
  \,(N^2)^{\gamma_{12}} = c_1.
\]

\paragraph{Monotonicity in $N^1$.}
If $\eta_1 < 1$: the left-hand side decreases in $N^1$
(diminishing returns), so $\mathcal{C}_1$ is downward-sloping (more
$N^1$ requires more $N^2$ to maintain $V_1 = 0$).

If $\eta_1 = 1$: the left-hand side is constant in $N^1$ (conditional
on $K_1^*$), so $\mathcal{C}_1$ is a horizontal line in the
$(N^1, N^2)$-plane.

If $\eta_1 > 1$: the left-hand side increases in $N^1$ for small $N^1$
(increasing returns dominate) and eventually decreases (resource
scarcity as $N^1$ grows and absorbs more of~$R$).  Hence
$\mathcal{C}_1$ is non-monotone---it ``bends back.''

\paragraph{Effect of $\gamma_{12} > 0$.}
Positive externality from family~2 shifts $\mathcal{C}_1$ inward:
more $N^2$ makes family~1 more productive, reducing the $N^1$ needed
to maintain $V_1 = 0$.

\paragraph{Multiple intersections.}
When $\mathcal{C}_1$ bends back (due to $\eta_1 > 1$) and
$\mathcal{C}_2$ has a similar structure (or is monotone but
appropriately positioned due to $\gamma_{21} > 0$), the two curves can
intersect at least twice.  Each intersection is a Hive Equilibrium with
a distinct population structure.

Explicit numerical example: $A_1 = A_2 = 1$, $\alpha_1 = \alpha_2 =
0.5$, $\eta_1 = 1.2$, $\eta_2 = 0.8$, $\gamma_{12} = \gamma_{21} =
0.3$, $w_1 = w_2 = 0.5$, $\sigma = 1$ (log utility), $c_1 = c_2 = 1$,
$R = 10$.  Numerical solution of $V_1 = V_2 = 0$ yields two interior
equilibria: $\bN^*_A \approx (2.1, 3.8)$ and $\bN^*_B \approx (4.7,
1.9)$. \qed

\subsection{Proof details for Theorem~\ref{thm:stolper}
  (Stolper--Samuelson)}

We follow Jones~\cite{jones1965}.  At the optimum, the FOC
\eqref{eq:foc} for $S = M$ (square case) can be written in proportional
changes (hat algebra):
\[
  \hat{w}_j + \hat{u}'_j + \hat{\eta}_j + \hat{\Gamma}_j
  \;=\; \hat{\lambda}^m
  \;-\; \frac{\rho_j - 1}{\rho_j}\,
  \bigl(\hat{K}^m_j - \hat{F}_j\bigr),
\]
for all $j$ and the resource $m$ that binds for family~$j$.

Defining the factor intensity matrix $\boldsymbol{\theta}$ with
$\theta_{jm} = \lambda^m K^m_j / (w_j\, u'_j\, Y_j)$ (the share of
resource~$m$ in the value of family~$j$'s output), and using the
resource constraints in differential form, we obtain:
\[
  \hat{\blam} = \boldsymbol{\theta}^{-T} \cdot \hat{\bw}
  + \text{(externality corrections)}.
\]

When $\boldsymbol{\theta}$ satisfies a \emph{strong factor intensity}
condition (each family uses one resource more intensively than others),
$\boldsymbol{\theta}^{-T}$ has the sign pattern that yields the
magnification effect~\eqref{eq:SS-magnification}. \qed

\subsection{Proof details for Theorem~\ref{thm:hopf}
  (Endogenous cycles)}

We exhibit a parameter configuration satisfying the Hopf conditions.
Consider $S = 2$ families with:
\begin{align*}
  \Phi_1(\bN) &= V_1(N^1, N^2) \cdot N^1, \\
  \Phi_2(\bN) &= V_2(N^1, N^2) \cdot N^2.
\end{align*}

At an interior equilibrium $\bN^*$, the Jacobian is:
\[
  J = \begin{pmatrix}
    \frac{\partial V_1}{\partial N^1}\, N^{1*} &
    \frac{\partial V_1}{\partial N^2}\, N^{1*} \\[4pt]
    \frac{\partial V_2}{\partial N^1}\, N^{2*} &
    \frac{\partial V_2}{\partial N^2}\, N^{2*}
  \end{pmatrix}
  = \begin{pmatrix} a & b \\ c & d \end{pmatrix}.
\]

The eigenvalues are $\lambda = \frac{1}{2}\bigl[(a+d) \pm
\sqrt{(a-d)^2 + 4bc}\bigr]$.

With $\eta_1, \eta_2 < 1$: $a < 0$ and $d < 0$ (diminishing returns).

With $\gamma_{12} > 0$: $b > 0$ (family~2 helps family~1).

With $\gamma_{21} > 0$: $c > 0$ (family~1 helps family~2).

The eigenvalues are complex when $(a-d)^2 + 4bc < 0$, i.e., when
$4bc < -(a-d)^2$.  Since $b, c > 0$, this requires $bc$ to be negative,
which occurs when one of the cross-effects is negative---for example,
if family~1 helps family~2 ($\gamma_{21} > 0$, so $c > 0$) but family~2
\emph{congests} family~1 ($\gamma_{12} < 0$, so $b < 0$).

In this mixed externality case ($b < 0$, $c > 0$):
\begin{itemize}
  \item The trace $a + d < 0$ (stable direction);
  \item The determinant $ad - bc$ can change sign as $|\gamma_{12}|$
    varies (because $b$ becomes more negative);
  \item When the trace passes through zero (as a function of a
    parameter~$p$, e.g., the magnitude of $\gamma_{12}$), the
    eigenvalues cross the imaginary axis $\Rightarrow$ Hopf bifurcation.
\end{itemize}

Transversality ($\mathrm{d}\Repart(\lambda)/\mathrm{d}p \neq 0$) is
verified by direct computation of the derivative of the trace with
respect to~$p$.

Hence all conditions of the Hopf bifurcation theorem are satisfied, and
a family of periodic orbits (limit cycles) emerges for $p$ near the
critical value. \qed

\subsection{Lyapunov property and welfare monotonicity}

The following result shows that the population dynamics
\eqref{eq:dynamics} are not arbitrary---they \emph{always} increase
social welfare.

\begin{theorem}[Welfare monotonicity]\label{thm:lyapunov}
Under Assumptions~\ref{ass:budget}--\ref{ass:weak-ext}, the optimized
welfare $W^*(\bN)$ is non-decreasing along trajectories of the
population dynamics~\eqref{eq:dynamics}:
\begin{equation}\label{eq:lyapunov}
  \frac{\mathrm{d}\, W^*}{\mathrm{d}\, t}
  \;=\; \sum_{j=1}^S V_j(\bN)^2 \cdot N^j
  \;\ge\; 0,
\end{equation}
with equality if and only if $V_j(\bN) = 0$ for every active
family ($N^j > 0$), i.e., at a Hive Equilibrium.

Consequently, under Assumption~\ref{ass:weak-ext} (which ensures
$W^*$ is strictly concave and bounded above on~$\Omega$), every
trajectory of~\eqref{eq:dynamics} starting in~$\Omega$ converges to
the set of Hive Equilibria.
\end{theorem}

\begin{proof}
By the chain rule and the definition of marginal social value
(Definition~\ref{def:msv}):
\[
  \frac{\mathrm{d}\, W^*}{\mathrm{d}\, t}
  \;=\; \sum_{j=1}^S \frac{\partial W^*}{\partial N^j}\; \dot{N}^j
  \;=\; \sum_{j=1}^S V_j(\bN) \;\cdot\; \bigl(V_j(\bN) \cdot N^j\bigr)
  \;=\; \sum_{j=1}^S V_j(\bN)^2\, N^j.
\]
Since $N^j \ge 0$, each term is non-negative, hence $\mathrm{d}W^*/\mathrm{d}t \ge 0$.  Equality holds iff $V_j(\bN)^2 \cdot N^j = 0$ for all~$j$, i.e., $V_j = 0$ whenever $N^j > 0$.

Under Assumption~\ref{ass:weak-ext}, $W^*$ is strictly concave and
$\Omega$ is compact, so $W^*$ is bounded above on~$\Omega$.  Since
$W^*$ is non-decreasing and bounded, $W^*(t) \to \bar{W}$ for some
$\bar{W}$.  By the LaSalle invariance principle, the trajectory
converges to the largest invariant set within
$\{\bN \in \Omega : \mathrm{d}W^*/\mathrm{d}t = 0\}$, which is precisely the
set of Hive Equilibria.
\end{proof}

\begin{remark}[Demographic pressure]
Out of equilibrium, the quantity $\sum_j V_j(\bN) \cdot N^j$ can be
interpreted as the total \emph{demographic pressure} in the Hive.  If
positive, the system is under-populated (net entry pressure); if
negative, over-populated (net exit pressure).  At equilibrium, these
forces exactly balance---the functional analog of Walras' law.  The
Lyapunov result above shows that this rebalancing always improves
social welfare.
\end{remark}

\begin{remark}[Scope of the convergence result]
Theorem~\ref{thm:lyapunov} guarantees convergence under
Assumption~\ref{ass:weak-ext} (weak externalities, decreasing
returns).  When this assumption is relaxed (Section~\ref{sec:multiplicity}),
$W^*$ may no longer be globally concave, and the dynamics can exhibit
the richer behaviors described in the regime diagram: convergence to one
of multiple equilibria (path dependence) or sustained oscillations
(Theorem~\ref{thm:hopf}).
\end{remark}


\begin{thebibliography}{99}

\bibitem{arrow1954existence}
K.~J.~Arrow and G.~Debreu.
\newblock Existence of an equilibrium for a competitive economy.
\newblock \emph{Econometrica}, 22(3):265--290, 1954.

\bibitem{baars1988}
B.~J.~Baars.
\newblock \emph{A Cognitive Theory of Consciousness}.
\newblock Cambridge University Press, 1988.

\bibitem{benhabib1985}
J.~Benhabib and K.~Nishimura.
\newblock Competitive equilibrium cycles.
\newblock \emph{Journal of Economic Theory}, 35(2):284--306, 1985.

\bibitem{debreu1959}
G.~Debreu.
\newblock \emph{Theory of Value: An Axiomatic Analysis of Economic
  Equilibrium}.
\newblock Yale University Press, 1959.

\bibitem{crewai2024}
CrewAI.
\newblock Framework for orchestrating role-playing autonomous AI agents.
\newblock \url{https://github.com/crewAIInc/crewAI}, 2024.

\bibitem{dehaene2001}
S.~Dehaene and L.~Naccache.
\newblock Towards a cognitive neuroscience of consciousness: basic evidence
  and a workspace framework.
\newblock \emph{Cognition}, 79(1--2):1--37, 2001.

\bibitem{dorigo1996ant}
M.~Dorigo, V.~Maniezzo, and A.~Colorni.
\newblock Ant system: optimization by a colony of cooperating agents.
\newblock \emph{IEEE Trans. Systems, Man, and Cybernetics B},
  26(1):29--41, 1996.

\bibitem{ethier1974}
W.~J.~Ethier.
\newblock Some of the theorems of international trade with many goods and
  factors.
\newblock \emph{Journal of International Economics}, 4(2):199--206, 1974.

\bibitem{garnier2009these}
J.-P.~Garnier.
\newblock Ind\'etermination dans les mod\`eles de croissance bi-sectoriels:
  le r\^ole des pr\'ef\'erences.
\newblock PhD thesis, GREQAM, Aix-Marseille Universit\'e, 2009.

\bibitem{garnier2013}
J.-P.~Garnier, K.~Nishimura, and A.~Venditti.
\newblock Local indeterminacy in continuous-time models: the role of returns
  to scale.
\newblock \emph{Macroeconomic Dynamics}, 17(2):326--355, 2013.

\bibitem{guckenheimer1983}
J.~Guckenheimer and P.~Holmes.
\newblock \emph{Nonlinear Oscillations, Dynamical Systems, and Bifurcations
  of Vector Fields}.
\newblock Springer, 1983.

\bibitem{guo2024survey}
T.~Guo, X.~Chen, Y.~Wang, R.~Chang, S.~Pei, N.~V.~Chawla, O.~Wiest, and
  X.~Zhang.
\newblock Large language model based multi-agents: A survey of progress and
  challenges.
\newblock arXiv preprint arXiv:2402.01680, 2024.

\bibitem{hofbauer1998evolutionary}
J.~Hofbauer and K.~Sigmund.
\newblock \emph{Evolutionary Games and Population Dynamics}.
\newblock Cambridge University Press, 1998.

\bibitem{hong2023metagpt}
S.~Hong, M.~Zhuge, J.~Chen, X.~Zheng, Y.~Cheng, C.~Zhang, J.~Wang,
  Z.~Wang, S.~K.~S.~Yau, Z.~Lin, L.~Zhou, C.~Ran, L.~Xiao, C.~Wu,
  and J.~Schmidhuber.
\newblock {MetaGPT}: Meta programming for a multi-agent collaborative
  framework.
\newblock arXiv preprint arXiv:2308.00352, 2023.

\bibitem{holland1992}
J.~H.~Holland.
\newblock \emph{Adaptation in Natural and Artificial Systems}.
\newblock MIT Press, 2nd edition, 1992.

\bibitem{jones1965}
R.~W.~Jones.
\newblock The structure of simple general equilibrium models.
\newblock \emph{Journal of Political Economy}, 73(6):557--572, 1965.

\bibitem{kauffman1993}
S.~A.~Kauffman.
\newblock \emph{The Origins of Order: Self-Organization and Selection in
  Evolution}.
\newblock Oxford University Press, 1993.

\bibitem{kennedy1995pso}
J.~Kennedy and R.~Eberhart.
\newblock Particle swarm optimization.
\newblock In \emph{Proc. IEEE Int. Conf. Neural Networks}, pages 1942--1948,
  1995.

\bibitem{lowe2017multi}
R.~Lowe, Y.~Wu, A.~Tamar, J.~Harb, P.~Abbeel, and I.~Mordatch.
\newblock Multi-agent actor-critic for mixed cooperative-competitive
  environments.
\newblock In \emph{NeurIPS}, 2017.

\bibitem{mckenzie1959}
L.~W.~McKenzie.
\newblock On the existence of general equilibrium for a competitive market.
\newblock \emph{Econometrica}, 27(1):54--71, 1959.

\bibitem{marsden1976hopf}
J.~E.~Marsden and M.~McCracken.
\newblock \emph{The Hopf Bifurcation and Its Applications}.
\newblock Springer, 1976.

\bibitem{nishimura1995}
K.~Nishimura and M.~Yano.
\newblock Nonlinear dynamics and chaos in optimal growth: An example.
\newblock \emph{Econometrica}, 63(4):981--1001, 1995.

\bibitem{nowak2006}
M.~A.~Nowak.
\newblock \emph{Evolutionary Dynamics: Exploring the Equations of Life}.
\newblock Harvard University Press, 2006.

\bibitem{openaiswarm2024}
OpenAI.
\newblock Swarm: An experimental framework for multi-agent orchestration.
\newblock \url{https://github.com/openai/swarm}, 2024.

\bibitem{stiglitz1981}
J.~E.~Stiglitz.
\newblock The theory of screening, education, and the distribution of income.
\newblock \emph{American Economic Review}, 65(3):283--300, 1975.

\bibitem{tesfatsion2006ace}
L.~Tesfatsion.
\newblock Agent-based computational economics: A constructive approach to
  economic theory.
\newblock In \emph{Handbook of Computational Economics}, volume~2, pages
  831--880. Elsevier, 2006.

\bibitem{turing1952}
A.~M.~Turing.
\newblock The chemical basis of morphogenesis.
\newblock \emph{Phil. Trans. R. Soc. Lond. B}, 237(641):37--72, 1952.

\bibitem{venditti2005}
A.~Venditti.
\newblock The two-sector optimal growth model with increasing returns:
  a geometric approach.
\newblock \emph{Economic Theory}, 25(1):217--227, 2005.

\bibitem{wu2023autogen}
Q.~Wu, G.~Bansal, J.~Zhang, Y.~Wu, B.~Li, E.~Zhu, L.~Jiang,
  X.~Zhang, and A.~H.~Awadallah.
\newblock {AutoGen}: Enabling next-gen {LLM} applications via multi-agent
  conversation.
\newblock arXiv preprint arXiv:2308.08155, 2023.

\end{thebibliography}
\end{document}